\documentclass[12pt]{article}
\usepackage{amsmath,amsfonts,amssymb,amsthm,graphicx}
\usepackage{authblk}

\newtheorem{theorem}{Theorem}[section]

\newtheorem{prop}[theorem]{Proposition}
\newtheorem{cor}[theorem]{Corollary}

\newcommand{\Ref}[1]{(\ref{#1})}

\newcommand{\ee}[1]{{\rm e}^{#1}}
\newcommand{\ii}{{\rm i}}

\newcommand{\tet}{\theta}

\newcommand{\Z}{{\mathbb Z}}

\newcommand{\cV}{{\mathcal  V}}
\newcommand{\cW}{{\mathcal  W}}
\newcommand{\cH}{{\mathcal  H}}
\newcommand{\cE}{{\mathcal  E}}

\newcommand{\cN}{{\mathcal  N}}
\newcommand{\cQ}{{\mathcal  Q}}
\newcommand{\cC}{{\mathcal  C}}

\newcommand{\tx}{\tilde{x}}
\newcommand{\tN}{\tilde{N}} 
\newcommand{\ty}{\tilde{y}}
\newcommand{\tM}{\tilde{M}} 
\newcommand{\tz}{\tilde{z}} 

\title{\Large{\bf{Source identity and kernel functions for Inozemtsev-type systems }}}

\date{\vspace{-1.0cm}\small \today\vspace{0.2cm}}

\author[1]{Edwin Langmann\footnote{langmann@kth.se}}
\author[2]{Kouichi Takemura\footnote{takemura@math.chuo-u.ac.jp}}
\affil[1]{Department of Theoretical Physics, Royal Institute of Technology KTH, SE-106 91 Stockholm, Sweden}
\affil[2]{Department of Mathematics, Chuo University, 1-13-27 Kasuga, Bunkyo-ku Tokyo 112-8551, Japan} 
\begin{document}

\maketitle

\begin{abstract}
The Inozemtsev  Hamiltonian is an elliptic generalization of the differential operator defining the BC$_N$ trigonometric quantum Calogero-Sutherland model, and its eigenvalue equation is a natural many-variable generalization of the Heun differential equation. 
We present kernel functions for  Inozemtsev Hamiltonians and Chalykh-Feigin-Veselov-Sergeev-type deformations thereof. Our main result is a solution of a heat-type equation for a generalized Inozemtsev Hamiltonian  which is the source for all these kernel functions. Applications are given, including a derivation of simple exact eigenfunctions and eigenvalues for the 
 Inozemtsev Hamiltonian.
\medskip

\noindent {MSC-class: 81Q05, 16R60}
\end{abstract}

\section{Introduction}
Integrable models in quantum mechanics are closely related to the mathematical theory of special functions. A famous example are Calogero-Moser-Sutherland models which describe an arbitrary number of identical particles moving in one dimension under the influence of particular one- and two-body potentials \cite{C,M,Su,OP}. The Hamiltonians of these models are differential operators that have eigenfunctions providing natural many-body generalizations of the classical orthogonal polynomials. For example, the so-called BC$_N$ trigonometric Calogero-Sutherland model has energy eigenfunctions given by many-variable Jacobi polynomials; see e.g.\ \cite{BF,HL}. 

The BC$_N$ trigonometric Calogero-Sutherland model has an elliptic generalization defined by the Hamiltonian 
\begin{equation} 
\label{I} 
\begin{split}
H_N = \sum_{j=1}^N\Bigl( -\frac{\partial^2}{\partial x_j^2} + \sum_{\nu =0}^3 g_\nu(g_\nu-1)\wp(x_j+\omega_\nu) \Bigr) \qquad \qquad \qquad \qquad \\+ \sum_{1\leq j<k\leq N} 2\lambda(\lambda-1) \left\{ \wp(x_j-x_k) + \wp(x_j +x_k) \right\} \\
\end{split}
\end{equation} 
with $\wp(x)$ the usual Weierstrass elliptic function with periods $2\omega_1$ and $2\omega_3$  and where we use the notation 
\begin{equation*}
\omega_0=0,\quad \omega_2=-\omega_1-\omega_3 ,
\end{equation*} 
here and in the following.\footnote{Our notation for elliptic functions is as by Whittaker and Watson \cite{WW}, except that $\omega_2$ is denoted by us as $\omega_3$; for the convenience of the reader we collect the definitions of the functions we use in Appendix~\ref{appA1}.} This Hamiltonian depends on the particle number $N$ and five coupling parameters $g_0$, $g_1$, $g_2$, $g_3$, $\lambda$. 
It defines the natural quantum-analogue of a classical Liouville integrable system first presented by Inozemtsev in \cite{I}, and we refer to this quantum-many body system as the {\em Inozemtsev model}. 
The integrability of the Inozemtsev model was partially established by van Diejen \cite{vD}, and Oshima \cite{O} described its commuting operators (higher-order Hamiltonians) completely. In the following we sometimes write $H_N\left(x;\{g_\nu\}_{\nu=0}^3,\lambda\right)$ for the Inozemtsev Hamiltonian in \Ref{I}, to indicate the argument $x=(x_1,\ldots,x_N)$ and the coupling parameters.  Note that this differential operator allows for a quantum mechanical interpretation only if one assumes that $\omega_1>0$, $-\ii\omega_3>0$, and that all coupling parameters  are real. However, many of our results hold true with lesser restrictions.  

The Inozemtsev model is interesting already in the one-variable case: the eigenvalue equation $H_1\psi(x)=E\psi(x)$ of the differential operator 
\begin{equation}
\label{Heun} 
H_1  = -\frac{\partial^2}{\partial x^2} + \sum_{\nu=0}^3 g_\nu(g_\nu-1)\wp(x + \omega_\nu)
\end{equation} 
is equivalent to the {\em Heun differential equation}, which is a standard form of a second-order Fuchsian differential equation with four singularities and a topic of current research in special function theory; see \cite{Ron,SL,TakID}. The Heun differential equation, and the differential equations of its confluent type, appear in several physics contexts, including quantum mechanics, general relativity, models of crystal imperfections \cite{SL}, and the AdS/CFT correspondence \cite{KSY}. In this paper we present generalizations of two known functional identities involving the Heun differential operator in \Ref{Heun}  and functions that are products of powers of Jacobi theta functions $\tet_{\nu +1} (x)$, $\nu =0,1,2,3$ (see Appendix~\ref{appA1} for precise definitions).  The first known identity is as follows: The function  
\begin{equation}
\label{Psi}
\Psi_1(x) = \prod_{\nu=0}^3\tet_{\nu +1} (x)^{g_\nu} 
\end{equation} 
obeys the equation
\begin{equation} 
\label{nsHeunPsi}
\left\{ 2(g_0 +g_1 +g_2 +g_3)  \frac{\partial}{\partial\beta} +H_1-C_1 \right\}\Psi_1(x) = 0 
\end{equation}
with a known constant $C_1$ (see \Ref{CN}), where $\beta= 2 \omega_1 \omega_3/(\pi\ii)$ ($\omega_1$ is fixed). This {\em non-stationary Heun equation} appears in several physics contexts, including the Wess-Zumino-Witten model \cite{EK}, the eight-vertex model \cite{BM} and Liouville field theory \cite{FLNO}. The second known identity provides a {\em kernel function} for a pair of Heun differential operators (recall that a function $F(x,y)$ of two variables $x$ and $y$ is called a kernel function of two differential operators $D(x)$ and $\tilde D(y)$ if  $[D(x)-\tilde D(y)-c]F(x,y)=0$ for some constant $c$): It is known that the function
\begin{equation} 
\label{F}
\Psi_{1,1}(x,y) = \frac{\prod_{\nu=0}^3\theta_{\nu+1}(x)^{g_\nu} \theta_{\nu+1}(y)^{\tilde g_\nu}}{\theta_1(x-y)^\lambda\theta_1(x+y)^\lambda}
\end{equation} 
is a kernel function of the Heun differential operators $H(x;\{g_\nu\}_{\nu=0}^3)$ and $H(y;\{\tilde g_\nu\}_{\nu=0}^3)$ provided that $\tilde g_\nu=\lambda-g_\nu$ and $\lambda=(g_0+g_1+g_2+g_3)/2$ \cite{Rui3,Tak3}.  In fact, we find that, for arbitrary $\lambda$, the function in \Ref{F} is a {\em generalized kernel function} of these differential operators in the following sense, 
\begin{equation} 
\label{kernel}
\left\{  2( g_0+g_1+g_2+g_3 -2\lambda)  \frac{\partial}{\partial\beta} +H_1(x;\{g_\nu\}_{\nu=0}^3) - H_1(y;\{\tilde g_\nu\}_{\nu=0}^3) +C_{1,1}  \right\}\Psi_{1,1}(x,y) =0 
\end{equation} 
with a known constant  $C_{1,1}$ (see \Ref{CNM}). Moreover, our results include another generalized kernel function for two Heun-type differential operators: the function
\begin{equation}
\label{F2} 
\tilde\Psi_1(x,y) = \theta_1(x-y)\theta_1(x+y) \prod_{\nu=0}^3\theta_{\nu+1}(x)^{g_\nu} \theta_{\nu+1}(y)^{\tilde g_\nu'} 
\end{equation} 
obeys
\begin{equation} 
\label{kernel2} 
\left\{ 2( g_0+g_1+g_2+g_3 +2)  \frac{\partial}{\partial\beta} +H_1(x;\{g_\nu\}_{\nu=0}^3) +\lambda H_1(y;\{\tilde g'_\nu\}_{\nu=0}^3) +\tilde C_{1,1}  \right\}\tilde \Psi_{1,1}(x,y) =0
\end{equation} 
provided that $\tilde g_\nu'=(2g_\nu +1-\lambda)/(2\lambda)$, for some known constant $\tilde C_{1,1}$ (see \Ref{tCNM}) and arbitrary $\lambda\neq 0$. 

In this paper we present and prove many-variable generalizations of the identities in the previous paragraph: we obtain a generalizations of the non-stationary Heun equation in  \Ref{Psi}--\Ref{nsHeunPsi}, and of the two kinds of generalized kernel function identities in \Ref{F}-\Ref{kernel} and \Ref{F2}--\Ref{kernel2}, to Inozemtsev Hamiltonians, for arbitrary particle numbers; see Corollaries~\ref{cor1}, \ref{cor2} and \ref{cor3}, respectively. 
Note that, in the latter two cases, the two Inozemtsev Hamiltonians 
 can have different particle numbers $N$ and $M$. The most general kernel function identity we obtain is for a pair of differential operators 
\begin{equation} 
\label{dI} 
\begin{split} 
H_{N,\tN}(x,\tx;\{g_\nu\}_{\nu=0}^3,\lambda) =& H_N(x;\{g_\nu\}_{\nu=0}^3,\lambda) -\lambda H_{\tN}(\tx;\{(\lambda+1-2g_{\nu})/(2\lambda)\}_{\nu=0}^3,1/\lambda)\\ & +\sum_{j=1}^N\sum_{k=1}^{\tN}2(1-\lambda)\left\{ \wp(x_j-\tx_k)+\wp(x_j+\tx_k)\right\} 
\end{split} 
\end{equation} 
and $H_{M,\tM}(y,\ty;\{\lambda -g_\nu\}_{\nu=1}^3,\lambda)$, with $H_N(x;\{g_\nu\}_{\nu=0}^3,\lambda)$ in \Ref{I}, for arbitrary particle numbers $N$, $\tN$, $M$, $\tM$; see Corollary~\ref{cor4}. All our results are special cases of a many-variable generalization of the non-stationary Heun equation in \Ref{Psi}--\Ref{nsHeunPsi} to a generalized Inozemtsev-type differential operator where all particle mass parameters can be different and where  the interaction strengths depend on these "masses" in a particular way; see Theorem~\ref{thm:main}. We refer to this as {\em source identity} since it is the source of all our other results: the latter are obtained in a simple way as special cases. Remarkably, direct proofs of these special cases are often more complicated than the proof of the source identity. 

It is known that integrable quantum mechanical models of Calogero-Moser-Sutherland type allow for deformations that share many of their beautiful mathematical properties \cite{CFV,Sergeev,SV}, and for the BC$_N$ trigonometric Calogero-Sutherland system this deformation corresponds to the trigonometric limit of the differential operator in \Ref{dI}. It thus is natural to conjecture that the generalization of this deformation to the Inozemtsev model is given by \Ref{dI}. 

The differential operator defining the Sutherland model (= $A_{N-1}$ trigonometric Calogero-Sutherland model) has a well-known kernel functions which can be used to construct the eigenfunctions and eigenvalues of the Sutherland model \cite{EL0}, and this approach can be generalized to the elliptic case \cite{EL} and to all quantum Calogero-Moser-Sutherland models associated with classical orthogonal polynomials  \cite{HL}. Moreover, the latter results allow for a natural generalization to the deformed models \cite{HL}. Our results in the present paper provide the starting point to generalize these results to the Inozemtsev mode and its deformation in \Ref{dI}. 

The first example of a source identity was found for the Sutherland model by Sen \cite{Sen}, and that this identities can be used to obtain kernel functions for the Sutherland model and its deformations was pointed out in \cite{EL1}. Source identities for all quantum Calogero-Moser-Sutherland models where obtained and used to derive kernel functions in \cite{HL}. A source identity allowing to derive kernel functions for the elliptic generalizations of the Sutherland model and their deformations was presented in \cite{EL2}. The present paper generalizes, to the elliptic case, results previously obtained in \cite{HL}. 

We mention four further topics for future research suggested by the results in the present paper. First, kernel functions of Calogero-Sutherland models can be regarded as a natural quantum analogue of B\"acklund transformations found by Wojciechowski \cite{Woj}; see also \cite{KS}. This suggests that our results can provide B\"acklund transformations for the  classical version of the Inozemtsev model. 
Second, kernel functions for the Sutherland model have been used to construct $Q$-operators that allow to derive integral representations for the eigenfunctions of this model \cite{KMS}. This suggests that it is possible to extend Sklyanin's separation-of-variable approach \cite{Skly} to the Inozemtsev model using our results.
Third, as pointed in \cite{GGR} (see also \cite{T}), some special cases of the Inozemtsev model are quasi-exactly solvable in the sense that the computation of a finite number, $m$, of eigenfunctions can be reduced to diagonalizing a $m\times m$ matrix (note that the Hamiltonian in Eqs.\ (2)--(4) in \cite{GGR} is identical with the  Inozemtsev Hamiltonian as in \Ref{I} for $\lambda=a$, $g_0=g_1=g_3=b$, $g_2=-3b-2[m+\lambda(N-1)]$). We find that these cases are special with regard to our kernel functions (since, by applying the result in Corollary~\ref{cor3} to the parameters above and $(N,M)=(N,m)$, we obtain $A_{N,M}=0$). This suggests that our results might shed new light, and possibly allow to extend, the results in \cite{GGR}. 
Fourth,  while many kernel functions for Ruijsenaars' relativistic generalizations of Calogero-Sutherland-type systems \cite{R}  are known \cite{KNS}, the results in \cite{HL} and the present paper suggest that many more such relativistic kernel functions should exist. We believe that a good strategy to find all such kernel identities would be to find relativistic generalizations of the source identities obtained in \cite{HL,EL2} and the present paper.

\section{Main result} 
As discussed, our main result is a heat-type equation for a Schr\"odinger-type differential operator. Our notations for elliptic functions is defined in Appendix~\ref{appA}. 

\begin{theorem}{ \bf (Source Identity):} 
\label{thm:main} 
For $\cN$ a positive integer, $\lambda$, $d_\nu$ ($\nu=0,1,2,3$) and $m_J\neq 0$ ($J=1,2,\ldots,\cN$) complex constants, and $X_J$ ($J=1,2,\ldots,\cN$) complex variables, let 
\begin{equation}
\label{Phi0} 
\Phi_0(X) =\left( \prod_{J=1}^{\cN} \prod_{\nu=0}^3 \theta_{\nu+1} (X_J)^{g_{\nu,J} } \right) \prod_{1\leq J<K\leq\cN} \theta_1(X_J-X_K)^{m_Jm_K\lambda}\theta_1(X_J+X_K)^{m_Jm_K\lambda}
\end{equation} 
and
\begin{equation}
\label{cH} 
\begin{split} 
\cH= \sum_{J=1}^N\frac1{m_J}\left( -\frac{\partial^2}{\partial X_J^2} + \sum_{\nu =0}^3 g_{\nu,J}(g_{\nu,J}-1) \wp(X_J+\omega_\nu ) \right) \qquad \qquad \qquad \qquad  \\ + \sum_{1\leq J<K\leq\cN} \gamma_{J,K}\left\{ \wp(X_J - X_K) +\wp(X_J+X_K)\right\} 
\end{split} 
\end{equation} 
with
\begin{equation}
\label{gammaJK} 
\gamma_{J,K} =\lambda(m_J+m_K)(\lambda m_Jm_K-1) 
\end{equation}
\begin{equation}
\label{gnuJ} 
g_{\nu,J} =m_J d_\nu +\frac{\lambda}{2}m_J ^2 . 
\end{equation} 
Then
\begin{equation}
\label{main result}
\Bigl\{ (4\lambda |m|+2|d|) \frac{\partial}{\partial\beta} +\cH -\cE_0\Bigr\} \Phi_0(X)=0 
\end{equation} 
with
\begin{equation}
\label{cE0} 
\begin{split} 
\cE_0= & (2 \lambda |m|+|d|) \left\{\cN- \lambda  (|m|^2 +|m^2| ) - |m| |d| \right\}\frac{\eta _1}{\omega _1}  \\
& + |m| \{ (d_0 d_1+ d_2d_3 )e_1 +(d_0 d_2+ d_1 d_3 )e_2 +(d_0 d_3+ d_1 d_2 )e_3 \} 
\end{split} 
\end{equation}
\begin{equation}
\label{nm} 
|d|= \sum_{\nu=0}^3 d_\nu ,\quad |m|= \sum_{J=1}^{\cN} m_J,\quad |m^2|= \sum_{J=1}^{\cN} m_J^2. 
\end{equation} 
\end{theorem} 

\begin{proof}
Consider the differential operator 
\begin{equation} 
\label{cHQQ}
\tilde{\cH}= \sum_{J} \frac{1}{m_J}\cQ_J^+\cQ_J^-
\end{equation} 
with
\begin{equation} 
\label{cQ} 
\cQ_J ^\pm = \pm\frac{\partial}{\partial X_J} + \cV_J,\quad \cV_J=\frac1{\Phi_0(X)}\frac{\partial\Phi_0(X)}{\partial X_J}.
\end{equation} 
Using identities of elliptic functions collected and proved in Appendix~\ref{appA2} we find, by straightforward computations (details are given in Appendix~\ref{appB}), 
\begin{equation} 
\label{result}
\tilde\cH = (4|m|\lambda+2|d|)\frac{1}{\Phi_0}\frac{\partial}{\partial\beta}\Phi_0 +\cH -\cE_0 .
\end{equation}
By definition, $\cQ^-_J\Phi_0=0$ for all $J$. Thus $\tilde\cH\Phi_0=0$, and \Ref{result}  implies our result in \Ref{main result}. 
\end{proof} 

Under suitable restrictions on parameters (see below), the differential operator $\cH$ in \Ref{cH} has a natural physical interpretation as Hamiltonian describing $\cN$ distinguishable quantum particles with interactions, and the results above provide the exact groundstate and groundstate energy of this Hamiltonian. Namely, if $\omega_1>0$, $-\ii\omega_3>0$, $\lambda>0$, $m_J>0$, $d_0>-\lambda m_J/2$ and $d_1 >-\lambda m_J/2$ for $J=1,2,\ldots,\cN$, $d_\nu$ real for $\nu=2,3$, and $4|m|\lambda+2|d|=0$, then the Hamiltonian in \Ref{cH} defines a unique self-adjoint operator on the Hilbert space $L^2([0,\omega_1]^\cN)$ which has $\Phi_0(X)$ as groundstate and $\cE_0$ as groundstate energy.  (This is true because, under these conditions,  $\cE_0$ and the potential terms in the Hamiltonian $\cH$  are real,  the function $\Phi_0(X)$ is square-integrable,  and the Hilbert space adjoint of $\cQ_J^-$ is equal to the closure of $\cQ_J^+$.   This, \Ref{cHQQ},  \Ref{result},  and the vanishing of $\partial/\partial\beta$-term in \Ref{result} imply that $\sum_J (\cQ^{-}_J)^\dag\cQ^{-}_J/m_J+\cE_0$ defines such a selfadjoint extension $\cH$; see e.g.\ \cite{RS2}). In the rest of this paper the self-adjointness of Inozemtsev-type differential operators will play no role. 

\section{Special cases} 
To state important special cases of our main result we use the following notation
\begin{equation} 
\label{I01} 
\begin{split}
H_N(x;\{g_\nu\}_{\nu=0}^3,\lambda) = & \sum_{j=1}^N\Bigl( -\frac{\partial^2}{\partial x_j^2} + \sum_{\nu =0}^3 g_\nu(g_\nu-1)\wp(x_j+\omega_\nu) \Bigr) \qquad\qquad\\ & + \sum_{1\leq j<k\leq N} 2\lambda(\lambda-1) \left\{ \wp(x_j-x_k) + \wp(x_j +x_k) \right\} , 
\end{split}
\end{equation} 
\begin{equation}
\Psi_N(x;\{g_\nu\}_{\nu=0}^3,\lambda) =\left(  \prod_{j=1}^{N} \prod_{\nu=0}^3 \theta_{\nu+1} (x_j)^{g_{\nu} } \right) \prod_{1\leq j<k\leq N} \theta_1(x_j-x_k)^{\lambda} \theta_1(x_j+x_k)^{\lambda} 
\end{equation}
for $x=(x_1,\ldots,x_N)$ and complex variables $x_j$. We also use the abbreviations
\begin{equation}
\label{c0} 
c_0 = \{ (g_0 g_1+ g_2g_3 )e_1 +(g_0 g_2+ g_1 g_3 )e_2 +(g_0 g_3+ g_1 g_2 )e_3 \} ,
\end{equation} 
\begin{equation}
|g|=g_0+g_1+g_2+g_3. 
\end{equation} 

We first state the many-variable generalization of the non-stationary Heun equation in \Ref{Psi}--\Ref{nsHeunPsi}. 

\begin{cor}
\label{cor1} 
For $N$ a positive integer, $g_\nu$ ($\nu=0,1,2,3$) and $\lambda$ complex parameters, the following holds true
\begin{equation} 
\left\{ A_N\frac{\partial}{\partial\beta} +H_N(x;\{g_\nu\}_{\nu=0}^3,\lambda) -C_N  \right\} \Psi_N(x)=0
\end{equation}
with
\begin{equation}
A_N=4\lambda(N-1) + 2|g|
\end{equation} 
\begin{equation}
\label{CN} 
C_N= \frac{A_N}{2}N[1-\lambda(N-1)-|g|]\frac{\eta_1}{\omega_1} + Nc_0. 
\end{equation}  
\end{cor} 

\begin{proof}
Set $\cN=N$, $d_\nu=g_\nu-\lambda/2$ ($\nu=0,1,2,3$), $(m_J,X_J)=(1,x_J)$ for $J=1,2,\ldots,N$ in Theorem~\ref{thm:main}, and rename $\cH$, $\Phi_0(X)$, $\cE_0$  to $H_N(x)$, $\Psi_N(x)$, $C_N$, respectively.
Recall $e_1+e_2+e_3=0$, which implies $ \{ (d_0 d_1+ d_2d_3 )e_1 +(d_0 d_2+ d_1 d_3 )e_2 +(d_0 d_3+ d_1 d_2 )e_3 \} =c_0$. 
\end{proof} 

The many-variable generalization of the generalized kernel function identity in \Ref{F}--\Ref{kernel} is as follows. 
 
\begin{cor}
\label{cor2}
For $N$, $M$ non-negative integers such that $N+M>0$, $g_\nu$ ($\nu=0,1,2,3$) and $\lambda$ complex parameters, let $\tilde g_\nu=\lambda-g_\nu$ and 
\begin{equation} 
\Psi_{N,M}(x,y) =\frac{\Psi_N(x;\{g_\nu\}_{\nu=0}^3,\lambda)\Psi_M(y;\{\tilde g_\nu\}_{\nu=0}^3,\lambda)}{\prod_{j=1}^N\prod_{k=1}^M\theta_1(x_j-y_k)^\lambda\theta_1(x_j+y_k)^\lambda}.
\end{equation} 
Then
\begin{equation}
\begin{split} 
\left\{ A_{N,M}\frac{\partial}{\partial\beta} +H_N(x;\{g_\nu\}_{\nu=0}^3,\lambda) - H_M(y;\{\tilde g_\nu\}_{\nu=0}^3,\lambda)-C_{N,M} \right\} \Psi_{N,M}(x,y)=0
\end{split} 
\end{equation} 
with
\begin{equation}
A_{N,M}=4\lambda(N-M-1)+2|g|
\end{equation} 
\begin{equation}
\label{CNM} 
C_{N,M} = \frac{A_{N,M}}{2}\{(N+M)(1-\lambda) -(N-M)[(N-M-2)\lambda+|g|]\}\frac{\eta_1}{\omega_1}+(N-M)c_0.
\end{equation} 
\end{cor} 
\begin{proof} 
Similarly as above, but now set $\cN=N+M$, $d_\nu=g_\nu-\lambda/2$ ($\nu=0,1,2,3$), and 
\begin{equation*}
(m_J, X_J) = \left\{ 
\begin{array}{lr}
(1 ,x_J) ,  & J=1, \dots, N \\
(-1,y _{J-N}), & J=N+1 , \dots , N+M  
\end{array}
\right.
\end{equation*}
in Theorem~\ref{thm:main}. 
\end{proof} 

The many-variable generalization of the generalized kernel function identity in \Ref{F2}--\Ref{kernel2} is as follows. 

\begin{cor}
\label{cor3}
For $N$, $M$ non-negative integers such that $N+M>0$, $g_\nu$ ($\nu=0,1,2,3$) and $\lambda\neq 0$ complex parameters, let $\tilde g_\nu'=(2g_\nu+1-\lambda)/(2\lambda)$ and 
\begin{equation} 
\tilde \Psi_{N,M}(x,y) =\left(\prod_{j=1}^N\prod_{k=1}^M\theta_1(x_j-y_k)\theta_1(x_j+y_k)
\right) \Psi_N(x;\{g_\nu\}_{\nu=0}^3,\lambda)\Psi_M(y;\{\tilde g'_\nu\}_{\nu=0}^3,1/\lambda) .
\end{equation} 
Then
\begin{equation}
\begin{split} 
\Biggl\{ \tilde A_{N,M}\frac{\partial}{\partial\beta} +H_N(x;\{g_\nu\}_{\nu=0}^3,\lambda) +\lambda H_M(y;\{g'_\nu\}_{\nu=0}^3,1/\lambda)-\tilde C_{N,M} \Biggr\} \tilde \Psi_{N,M}(x,y)=0
\end{split} 
\end{equation} 
with
\begin{equation}
\label{tANM} 
\tilde A_{N,M} = 4\lambda(N-1) +4M+2|g|
\end{equation} 
\begin{equation}
\label{tCNM} 
\tilde C_{N,M} = \frac{\tilde A_{N,M}}{2}\{N+M-(N+M/\lambda)[(N-2)\lambda+M+|g|]-N\lambda-M/\lambda\}\frac{\eta_1}{\omega_1} +(N+M/\lambda)c_0.
\end{equation} 
\end{cor} 
\begin{proof} 
Similarly as above, but now set $\cN=N+M$, $d_\nu=g_\nu-\lambda/2$ ($\nu=0,1,2,3$), and 
\begin{equation*}
(m_J, X_J) = \left\{ 
\begin{array}{lr}
(1 ,x_J) ,  & J=1, \dots, N \\
(1/\lambda ,y _{J-N}), & J=N+1 , \dots , N+M  
\end{array}
\right.
\end{equation*}
in Theorem~\ref{thm:main}. 
\end{proof} 

We finally state the generalized kernel function identity for deformed Inozemtsev Hamiltonians. 
Note that all previous results stated in this section are special cases of this. 

\begin{cor}
\label{cor4}
For $N$, $\tN$, $M$, $\tM$ non-negative integers such that $N+\tN+M+\tM>0$, $d_\nu$ ($\nu=0,1,2,3$) and $\lambda\neq 0$ complex parameters, let $g_\nu=d_\nu+\lambda/2$ ($\nu=0,1,2,3$)  and 
\begin{equation} 
\label{HNtN1} 
\begin{split}
H^{(\pm )}_{N,\tN}(x,\tx) = & H_N(x;\{\lambda/2\pm d_\nu\}_{\nu=0}^3,\lambda) - \lambda \tilde{H}_{\tN} (\tx;\{ (1/2 \mp d_{\nu } )/\lambda \}_{\nu=0}^3,1/\lambda) \\
& +\sum_{j=1}^N\sum_{k=1}^{\tN} 2(1-\lambda) \left\{ \wp(x_j-\tx _k) +  \wp(x_j+\tx _k) \right\} ,
\end{split}
\end{equation} 
\begin{equation}
\label{PsiNtN1} 
 \Psi_{N,\tN } ^{(\pm )}(x,\tx ) = \frac{\Psi_{N}(x;\{\lambda/2\pm d_\nu\}_{\nu=0}^3,\lambda) \Psi_{\tN}(\tx;\{ (1/2 \mp d_{\nu } )/\lambda \}_{\nu=0}^3,1/\lambda)}{\prod _{j=1}^N \prod _{k=1}^{\tN} \theta_1(x_j-\tx_k) \theta_1(x_j+\tx_k)} ,
\end{equation}
\begin{equation}
\label{PsiNtNMtM} 
\begin{split}
\Psi_{N,\tN,M,\tM}(x,\tx,y,\ty) = \Psi_{N,\tN } ^{(+ )}(x,\tx ) \Psi_{M,\tM } ^{(- )}(y,\ty )   \qquad \qquad \qquad \qquad \\
\times \prod_{r=\pm} \left( \prod_{j=1}^N \frac{\prod_{k=1}^{\tM } \theta_1(x_j-r \ty _k)}{\prod_{k=1}^M \theta_1(x_j-r y_k)^{\lambda}} \right)\left(\prod_{j=1}^{\tN} \frac{\prod_{k=1}^M \theta_1(\tx_j-r y_k)}{\prod_{k=1}^{\tM } \theta_1(\tx_j-r\ty_k)^{1/\lambda}} \right).
\end{split}
\end{equation}
Then 
\begin{equation}
\label{IdNtNMtM} 
\Bigl(  A_{N,\tN,M,\tM}  \frac{\partial}{\partial\beta} +H_{N,\tN}^{(+)}(x,\tx)-H^{(-)}_{M,\tM}(y,\ty) - C_{N,\tN,M,\tM}  \Bigr)\Psi_{N,\tN,M,\tM}(x,\tx,y,\ty) = 0 
\end{equation} 
with 
\begin{equation}
\label{ANtNMtM} 
A_{N,\tN,M,\tM} = 4\lambda(N-M-1) -4(\tN-\tM) + 2|g| 
\end{equation} 
\begin{equation}
\label{CNtNMtM} 
\begin{split} 
C_{N,\tN,M,\tM} =\frac{A_{N,\tN,M,\tM}}{2}\left\{N+\tN+M+\tM -|m|[(|m|-2)\lambda+|g|] -|m^2|\lambda\right\} \frac{\eta_1}{\omega_1}+ |m|c_0 , \\
|m|=N-M-(\tN-\tM)/\lambda ,\quad |m^2|=N+M+(\tN+\tM)/\lambda ^2. \qquad\qquad
\end{split} 
\end{equation} 
\end{cor} 

\begin{proof} 
Similarly as above, but now set $\cN=N+\tN+M+\tM$ and 
\begin{equation*}
(m_J, X_J) = \left\{ 
\begin{array}{lr}
(1 ,x_J) ,  & J=1, \dots, N \\
(-1/\lambda ,\tx _{J-N}) , & J=N+1 , \dots N+\tN \\
(-1 ,y_{J-N-\tN}) ,  & J=N+\tN +1, \dots, N+\tN +M \\
(1/\lambda ,\ty _{J-N-\tN -M}), & J=N+\tN +M+1 , \dots N+\tN +M+\tM 
\end{array}
\right.
\end{equation*}
in Theorem~\ref{thm:main}. 
\end{proof} 

Note that $H_{N,\tN}^{(+)}(x,\tx)$ in \Ref{HNtN1} is equal to $H_{N,\tN}(x,\tx;\{g_\nu\}_{\nu=0}^3,\lambda)$ in \Ref{dI}, and $H_{N,\tN}^{(-)}(x,\tx)$ in \Ref{HNtN1} is equal to $H_{N,\tN}(x,\tx;\{\lambda-g_\nu\}_{\nu=0}^3,\lambda)$. Thus the general kernel identity in Corollary~\ref{cor4} is a natural generalization of the one in Corollary~\ref{cor2}. 

We also note that the generalized kernel function identity on Corollary~\ref{cor4} is invariant under the following transformations,
\begin{equation}
\label{sym1} 
(N,\tN,M,\tM,\{g_\nu\}_{\nu=0}^3,\lambda)\to (M,\tM,N,\tN,\{\lambda-g_\nu\}_{\nu=0}^3,\lambda)
\end{equation} 
(the variable names should also be changed correspondingly, of course) and 
\begin{equation}
\label{sym2} 
(N,\tN,M,\tM,\{g_\nu\}_{\nu=0}^3,\lambda)\to (\tN,N,\tM,M,\{(\lambda+1-2g_\nu)/(2\lambda)\}_{\nu=0}^3,1/\lambda),
\end{equation} 
and these symmetries provide non-trivial checks of our computations. To be more specific: under the transformation in \Ref{sym1}, the constants in \Ref{ANtNMtM} and \Ref{CNtNMtM} change as $A_{N,\tN,M,\tM} \to - A_{N,\tN,M,\tM} $, 
$C_{N,\tN,M,\tM} \to - C_{N,\tN,M,\tM} $, and under the transformations in \Ref{sym2} they change as $A_{N,\tN,M,\tM} \to - A_{N,\tN,M,\tM}/\lambda$, $C_{N,\tN,M,\tM} \to - C_{N,\tN,M,\tM}/\lambda$, consistent with the transformation properties of the r.h.s.\ of the generalized kernel function identity in \Ref{IdNtNMtM}. Note that not only \Ref{sym1} but also \Ref{sym2} is a duality transformation (i.e.\ applying each of these transformations twice gives the identity).

\section{Applications}
\label{sec4} 
Using the kernel functions obtained in the previous section it is possible to extend methods developed in \cite{HL,EL,Tak1} (e.g.)  to construct eigenfunctions and eigenvalues of Inozemtsev-type differential operators. This section describes a general strategy and two simple examples. More systematic studies are left to future work. 

\subsection{Integral transformations}
\label{sec4.1} 
We explain how the kernel functions obtained in the previous section can be used to construct integral transformations that map a known generalized eigenfunction of a Inozemtsev-type differential operator to a generalized eigenfunction of another such operator. 

Let $\Psi_{N,\tN,M,\tM}\equiv \Psi_{N,\tN,M,\tM}(x,\tx,y,\ty)$ be a generalized kernel function, $H_{N,\tilde N}^{(+)}(x,\tilde x)$ and $H_{M,\tilde M}^{(-)}(y,\tilde y)$ Inozemtsev-type differential operators, and  $A_{N,\tN,M,\tM}$ and $C_{N,\tN,M,\tM} $ constants as in Corollary~\ref{cor4}. If $f(y,\tilde y)$ is a generalized eigenfunctions of the differential operator $H^{(-)}_{M,\tM}(y,\ty)$ in the following sense,  
\begin{equation}
\label{fytyH-} 
\left(A_{N,\tN,M,\tM} \frac{\partial}{\partial\beta} + H^{(-)}_{M,\tM}(y,\ty)  -E \right)f(y,\ty )=0
\end{equation}
for some constant $E$, then \Ref{IdNtNMtM} implies 
\begin{equation}
\label{eq} 
\begin{split} 
\left\{ A_{N,\tN,M,\tM} \frac{\partial}{\partial\beta} + H_{N,\tilde N}^{(+)}(x,\tx) - E-C_{N,\tN,M,\tM}\right\} \Psi_{N,\tN,M,\tM}f(y,\tilde y) \\ =  f(y,\ty )\{ H_{M,\tilde M}^{(-)}(y,\ty )\Psi_{N,\tN,M,\tM}\}  - \Psi_{N,\tN,M,\tM}\{H_{M,\tilde M}^{(-)}(y,\ty )f(y,\ty ) \} \\=
-\sum_{j=1}^{M} \frac{\partial}{\partial y_j} \left(f(y,\ty )\frac{\partial}{\partial y_j}\Psi_{N,\tN,M,\tM}  - \Psi_{N,\tN,M,\tM}\frac{\partial}{\partial y_j}f(y,\ty )  \right)  \\+ \lambda \sum_{k=1}^{\tM} \frac{\partial}{\partial \tilde y_k} \left(f(y,\ty )\frac{\partial}{\partial \tilde y_k}\Psi_{N,\tN,M,\tM}  - \Psi_{N,\tN,M,\tM}\frac{\partial}{\partial \tilde y_k}f(y,\ty )  \right).
\end{split} 
\end{equation} 
Integrating this with respect to the variables $(y,\ty)$ over a suitable region $\cC$, one finds that 
\begin{equation}
\label{tf}
\tilde f(x,\tx) =   \int_{\cC} \Psi_{N,\tN,M,\tM}(x,\tx,y,\ty)f(y,\tilde y)d^Myd^{\tM}\ty
\end{equation} 
is a generalized eigenfunctions of the differential operator $H_{N,\tilde N}^{(+)}(x,\tilde x)$, i.e., 
\begin{equation}
\left(A_{N,\tN,M,\tM} \frac{\partial}{\partial\beta} + H^{(+)}_{N,\tN}(x,\tx) -E-C_{N,\tN,M,\tM} \right)\tilde f(x,\tx )=0.
\end{equation} 
Note that a key point in the derivation of this result is that the region $\cC$ is suitable in the following sense: First, the integral in \Ref{tf} has to be well-defined, and second, the integral over the total derivative terms in the last two lines of \Ref{eq} must vanish (in general, Stokes' theorem implies that the latter is equal to an integral over the boundary of $\cC$). 

\subsection{Example 1}
\label{sec4.2} 
To be specific we assume throughout this section that $\omega_1>0$, $-\ii\omega_3>0$, and that $x_j$ and $\tx_j$ are real variables. As discussed, in this case the Inozemtsev Hamiltonian can be interpreted as a quantum mechanical model of a many-particle system. 

We consider the result in the previous section in the case $M=1$, $\tM=0$, and $\tilde{g}_\nu\in\{ 0,1\}$ ($\nu=0,1,2,3$). Then  $H^{(-)}_{M,\tM}(y,\ty)=-\partial^2/\partial y^2$, and it is trivial to find solutions of \Ref{fytyH-}: $f(y) = \exp(-\ii py)$ and $E=p^2$, with $p$ an arbitrary constant. Moreover, 
\begin{equation}
\Psi_{N,\tN,1,0}(x,\tx,y) = \Psi_{N,\tN } ^{(+ )}(x,\tx )\left(\prod_{\nu=0}^3\theta_{1+\nu}(y)^{\tilde{g}_\nu}\right)
\frac{\prod_{j=1}^{\tN}\theta_1(\tx_j-y)\theta_1(\tx_j+y)}{\prod_{j=1}^N\theta_1(x_j-y)^{\lambda}\theta_1(x_j+y)^{\lambda}}.
\end{equation}
A suitable integration region in this case is any path $\cC$ in the complex $y$-plane such that $\Psi_{N,\tN,1,0}(x,\tx,y)f(y)$ is analytic in some neighborhood $\cC$ and such that 
\begin{equation}
\label{suitable} 
\int_{\cC} \frac{\partial}{\partial y} \left(f(y)\frac{\partial}{\partial y}\Psi_{N,\tN,1,0}(x,\tx,y)  - \Psi_{N,\tN,1,0}(x,\tx,y)\frac{\partial}{\partial y}f(y)  \right) dy=0.
\end{equation} 
To find such a path we use that the Theta-functions $\theta_{1+\nu}(x)$ can be expressed in terms of meromorphic functions $\check\theta_{1+\nu}(z)$ of the variable $z=\exp(-2\ii x/R)$ as follows
\begin{equation} 
\begin{array}{ll}
\bigskip
\theta_1(x)= \ee{-\pi\ii}\ee{\tau\pi\ii/4}\ee{\ii x/R}\check\theta_1(z) , &\check\theta_1(z)  = \sum_{n=0}^\infty (-1)^n \ee{\tau\pi\ii n(n+1)}(z^{-n}-z^{n+1})\\
\bigskip
\theta_2(x)= \ee{\tau\pi\ii/4}\ee{\ii x/R}\check\theta_2(z) , &\check\theta_2(z) = \sum_{n=0}^\infty \ee{\tau\pi\ii n(n+1)}(z^{-n}+z^{n+1})\\
\bigskip
\theta_3(x)= \check\theta_3(z) , &\check\theta_3(z)  = 1+\sum_{n=1}^\infty \ee{\tau\pi\ii n^2}(z^{-n}+z^{n})\\
\medskip
\theta_4(x)= \check\theta_4(z) , &\check\theta_4(z)  = 1+\sum_{n=1}^\infty (-1)^{n-1} \ee{\tau\pi\ii n^2}(z^{-n}+z^{n})
\end{array} 
\end{equation} 
(this is a simple consequence of the definition of the Theta-functions in \Ref{thenu}). For $p=(2n+\tilde{g}_0+\tilde{g}_1)/R$ and $n$ an integer, one thus finds that $\Psi_{N,\tN,1,0}(x,\tx,y)\exp(-\ii py)$ is a holomorphic function in the variable $\xi=\exp(-2\ii y/R)$ in the region $1<|\xi|<q^{-2}$, 
and that \Ref{suitable} is satisfied for $\cC$ the straight line from $y=\ii \epsilon$ to $\pi R+\ii \epsilon$, with $\epsilon>0$ such that $\exp(2\epsilon/R)<q^{-2}$ (since this corresponds to a closed path in the complex $\xi$-plane where the integrand is holomorphic). 
Taking the limit $\epsilon\to 0^+$ we obtain the following.  
\begin{prop}
For $N$, $\tN$ non-negative integers such that $N+\tN>0$, $\lambda$ a non-zero constant, $\tilde{g}_\nu\in\{0,1\}$ ($\nu=0,1,2,3$), and $n$ an arbitrary integer, let $d_\nu=\lambda/2-\tilde{g}_\nu$, $H^{(+)}_{N,\tN}(x,\tx)$ as in \Ref{HNtN1}, and $\Psi^{(+)}_{N,\tN}(x,\tx)$ as in \Ref{PsiNtN1}. Then the function 
\begin{equation} 
\begin{split} 
\label{tilde fn} 
\tilde f_n(x,\tx) =  \Psi_{N,\tN } ^{(+ )}(x,\tx )&\lim_{\epsilon\to 0^+}\int_{\ii \epsilon}^{\pi R+\ii \epsilon} \left(\prod_{\nu=0}^3\theta_{1+\nu}(y)^{\tilde{g}_\nu}\right)\\  & \times 
\frac{\prod_{j=1}^{\tN}\theta_1(\tx_j-y)\theta_1(\tx_j+y)}{\prod_{j=1}^N\theta_1(x_j-y)^{\lambda}\theta_1(x_j+y)^{\lambda}}
\ee{-\ii(2n+\tilde{g}_0+\tilde{g}_1)y/R}dy
\end{split}  
\end{equation} 
is well-defined and obeys the equation
\begin{equation} 
\label{result1}
\left(A_{N,\tN,1,0}\frac{\partial}{\partial\beta} + H^{(+)}_{N,\tN}(x,\tx) - (2n+\tilde{g}_0+\tilde{g}_1)^2/R^2 -C_{N,\tN,1,0}  \right)\tilde f_n(x,\tx) =0  
\end{equation} 
with
\begin{equation}
A_{N,\tN,1,0} = 4\lambda N -4\tN - 2|\tilde{g}| ,
\end{equation} 
\begin{equation}
\begin{split} 
C_{N,\tN,1,0} =\frac{A_{N,\tN,1,0}}{2}\left\{N+\tN+1-|m|[(|m|+2)\lambda-|\tilde{g}|] -|m^2|\lambda\right\}+|m|c_0,\qquad \qquad\\ c_0=\{ (\tilde{g}_0 \tilde{g}_1+ \tilde{g}_2\tilde{g}_3 )e_1 +(\tilde{g}_0 \tilde{g}_2+ \tilde{g}_1 \tilde{g}_3 )e_2 +(\tilde{g}_0 \tilde{g}_3+ \tilde{g}_1 \tilde{g}_2 )e_3 \} , \qquad\qquad \quad\\
|\tilde{g}|=\tilde{g}_0+\tilde{g}_1+\tilde{g}_2+\tilde{g}_3,\quad |m|=N-1-\tN/\lambda ,\quad |m^2|=N+1+\tN/\lambda ^2. 
\end{split} 
\end{equation} 
\end{prop}
A noteworthy special case is $N=1$, $\tN=0$, and $\lambda=|\tilde{g}|/2$. Then the function in \Ref{tilde fn} is 
\begin{equation} 
\tilde f_n(x) = \left(\prod_{\nu=0}^3\theta_{1+\nu}(x)^{g_\nu}\right)\int_0^{\pi R} \frac{\left(\prod_{\nu=0}^3\theta_{1+\nu}(y)^{\tilde{g}_\nu}\right)}{\theta_1(x-y)^{\lambda}\theta_1(x+y)^{\lambda}}\ee{-\ii(2n+\tilde{g}_0+\tilde{g}_1)y/R}dy
\end{equation} 
with $g_\nu =|\tilde{g}|/2-\tilde{g}_\nu$, and  \Ref{result1} is the Heun differential equation
\begin{equation}
\left( -\frac{\partial^2}{\partial x^2} +\sum_{\nu=0}^{3}g_\nu(g_{\nu}-1) \wp(x+\omega_\nu) - (2n+\tilde{g}_0+\tilde{g}_1)^2/R^2  \right)\tilde f_n(x) =0 .
\end{equation} 

It is worth noting that the integral factor on the r.h.s.\  in \Ref{tilde fn} is proportional to the function $f_n(z,\tilde z)$, $z_j=\exp(-2\ii x_j/R)$ and $\tilde z_j=\exp(-2\ii \tilde x_j/R)$, defined by the following generating function, 
\begin{equation} 
\label{gen fun} 
\left(\prod_{\nu=0}^3\check\theta_{1+\nu}(\xi)^{\kappa_\nu}\right)
\frac{\prod_{j=1}^{\tN}\check\theta_1(\tz_j/\xi)\check\theta_1(\tz_j \xi)}{\prod_{j=1}^N\check\theta_1(z_j/\xi)^{\lambda}\check\theta_1(z_j \xi)^{\lambda}} = \sum_{n=-\infty}^\infty f_{n}(z,\tz) \xi^{-n}, 
\end{equation} 
where the series on the r.h.s.\ is absolutely convergent in the region $1<|\xi|<q^{-2}$. 

\subsection{Example 2}
We show how kernel functions can be used to transform Bethe ansatz solutions of the Heun equation obtained in \cite{Tak1} to eigenfunctions of Inozemtsev-type differential operators with arbitrary particle numbers $N$, $\tN$ (this is a generalization of a result obtained in \cite{TakID} for the case $N=1$, $\tN=0$). 
 
 As shown in \cite{Tak1}, for arbitrary non-negative integers $n_\nu$ ($\nu=0,1,2,3$) and arbitrary constants $\tilde E$, the Heun differential equation
 \begin{equation} 
 \label{HeunEq}
 \left( -\frac{\partial^2}{\partial y^2} + \sum_{\nu=0}^3 n_\nu(n_\nu+1)\wp(y+\omega_\nu)-\tilde E\right)f(y)=0 
 \end{equation} 
 has a non-zero solution which can be written as 
 \begin{equation}
f(y)=\frac{\displaystyle \exp(\kappa y)\prod_{j=1}^{n_0 +n_1 +n_2 +n_3 } \theta _1 (y+t_j)}{\theta _1 (y)^{n_0 }\theta _2 (y)^{n_1} \theta _3(y)^{n_2}\theta _4(y)^{n_3}},
\label{tBV}
\end{equation}
for some constants $t_j$ $(j=1,\dots ,n_0 +n_1 +n_2 +n_3)$ and $\kappa$ (see \cite{Tak1} for how these constants can be determined). 

Choosing $M=1$, $\tM=0$, $\tilde{g}_\nu\in\{n_\nu +1,-n_\nu\}$, $\lambda$ such that $A_{N,\tN,1,0}=0$, $\tilde E=E+C_{N,\tN,1,0}$, we can use the result in Section~\ref{sec4.1} to transform this solution to an eigenfunction of the Inozemtsev-type differential operator $H^{(+)}_{N,\tN}(x,\tx)$ with $g_\nu=\lambda-\tilde g_\nu$ ($\nu=0,1,2,3$). As a suitable integration region we now choose a closed path $\cC$ in the complex $y$-plane such that $\Psi_{N,\tN,1,0}(x,\tx,y)f(y)$ is analytic in some neighborhood of this path and such that \Ref{suitable} is fulfilled. 
For example, for $p\in \omega_1\Z+\omega_3\Z$, $i\in\{1,2,\ldots,N\}$, we can choose as $\cC$ a figure-eight contour in the $y$-plane which encloses $y=x_i$ counterclockwise, $y=2p - x_i $ clockwise, and which does not contain branching points of $\Psi_{N,\tN,1,0}(x,\tx,y)$ other than $y=x_i , 2p - x_i $ inside\footnote{The second author apologizes for a vague description of the integral routes in \cite[Proposition 6]{TakID}. The orientations of them should be specified as the figure-eight contour in this paper.}; see Fig.~A (the condition \Ref{suitable} is satisfied, because the function $ f(y)\frac{\partial}{\partial y}\Psi_{N,\tN,1,0}(x,\tx,y)- \Psi_{N,\tN,1,0}(x,\tx,y) \frac{\partial}{\partial y} f(y) $ comes back to the original value after analytic continuation along $\cC$). 

\begin{center}
\begin{picture}(40,80)(0,0)
\put(90,70){\circle*{3}}
\put(140,70){\circle*{3}}
\put(100,50){\circle*{3}}
\put(60,30){\circle*{3}}
\qbezier(90,70)(120,80)(140,80)
\qbezier(150,70)(150,80)(140,80)
\qbezier(150,70)(150,60)(140,60)
\qbezier(90,70)(120,60)(140,60)
\qbezier(90,70)(70,20)(60,20)
\qbezier(90,70)(50,40)(50,30)
\qbezier(60,20)(50,20)(50,30)
\put(140,80){\vector(1,0){1}}
\put(138,60){\vector(-1,0){1}}
\put(73,33){\vector(1,2){5}}
\put(66,51){\vector(-1,-1){10}}
\put(102,48){$p$}
\put(61,32){$x_i$}
\put(151,66){$2p-x_i$}
\end{picture}
Fig.~A. Figure-eight contour $\cC$.
\end{center}

Obviously there are many more such closed paths $\cC$, which we call suitable. We thus obtain the following. 

\begin{prop}
\label{res3}
For $N$, $\tN$ non-negative integers such that $N+\tN>0$, $n_\nu$ non-negative integers ($\nu=0,1,2,3$), let $\tilde{g}_\nu\in\{n_\nu +1,-n_\nu\}$,  $\lambda$ such that $2\lambda N-2\tN=\tilde{g}_0+\tilde{g}_1+\tilde{g}_2+\tilde{g}_3$, $d_\nu=\lambda/2-\tilde{g}_\nu$, $H^{(+)}_{N,\tN}(x,\tx)$ as in \Ref{HNtN1}, and $\Psi^{(+)}_{N,\tN}(x,\tx)$ as in \Ref{PsiNtN1}. Then, for any constant $E$, there exit constants $t_i$($i=1,2,\ldots,n_0+n_1+n_2+n_3$) and $\kappa$ such that, for any suitable closed path $\cC$ in the complex $y$-plane (e.g. the figure-eight contour in Fig. A), the function 
\begin{equation}
\begin{split}
& \Psi_{N,\tN } ^{(+ )}(x,\tx ) \int_{\cC} 
\frac{\prod_{j=1}^{\tN}\theta_1(\tx_j-y)\theta_1(\tx_j+y)}{\prod_{j=1}^N\theta_1(x_j-y)^{\lambda}\theta_1(x_j+y)^{\lambda}} \\
& \qquad \qquad \qquad \qquad \exp(\kappa y) \left(\prod_{\nu=0}^3\theta_{1+\nu}(y)^{ \tilde{g}_\nu - n_\nu}\right) \prod_{j=1}^{n_0 +n_1 +n_2 +n_3 } \theta _1 (y+t_j)dy , 
\end{split}
\end{equation}
is an eigenfunction of $H_{N,\tN}^{(+)}(x,\tx) $ with the eigenvalue $E$.
\end{prop}

Note that the constants $t_j$ and $\kappa$ can be specified by the condition that the function \Ref{tBV} satisfies \Ref{HeunEq} with the eigenvalue $\tilde E=E+C_{N,\tN,1,0}$. Note also that there is another expression of solutions of the Heun equation \Ref{HeunEq} by the Hermite-Krichever ansatz \cite{Tak4}, and we can obtain a similar result to Proposition~\ref{res3}.

We now describe a classical example of the Bethe ansatz where the constants $t_j$ and $\kappa$ can be specified in a simple manner. If $n_0=1$, $n_1=n_2=n_3=0$,  \Ref{HeunEq} reduces to a Lam\'e equation, and its solution in \Ref{tBV} can be written as
\begin{equation}
f(y) = \exp \left( -\zeta (t) y \right) \frac{\sigma (y+t)}{\sigma (y)} = \exp \left( \eta _1 t^2/(2 \omega _1) \right)  \exp \left(-\phi _1 (t) y  \right) \frac{\theta _1 (y+t)}{\theta _1 (y)}  
\label{LalphaLame1}
\end{equation}
with $t$ such that $\wp (t)=-\tilde{E}$; see \cite{Tak4} and references therein. By specializing to this case and $\tilde g_0=-1$, $\tilde g_1=\tilde g_2=\tilde g_3=0$ we obtain the following (note that $c_0=0$ if $\tilde g_1=\tilde g_2=\tilde g_3=0$).

\begin{prop}
\label{res4}
For $N$ a positive integer, $\tN$ a non-negative integer, let $\lambda=(2\tN+1)/(2N)$, $d_0=\lambda/2+1$,  $d_1=d_2=d_3=\lambda/2$, $H^{(+)}_{N,\tN}(x,\tx)$ as in \Ref{HNtN1}, and $\Psi^{(+)}_{N,\tN}(x,\tx)$ as in \Ref{PsiNtN1}. Then, for any closed suitable path $\cC$ in the complex $y$-plane as described above (e.g. the figure-eight contour in Fig. A), and for any constant $t$ such that $\wp(t)$ is finite, the function
\begin{equation}
\Psi_{N,\tN } ^{(+ )}(x,\tx ) \int_{\cC} 
\frac{\prod_{j=1}^{\tN}\theta_1(\tx_j-y)\theta_1(\tx_j+y)}{\prod_{j=1}^N\theta_1(x_j-y)^{\lambda}\theta_1(x_j+y)^{\lambda}} \exp(-\phi _1 (t)  y) \frac{ \theta _1 (y+t)}{ \theta_{1}(y)^{2}}dy 
\end{equation}
is an eigenfunction of $H_{N,\tN}^{(+)}(x,\tx) $ with the eigenvalue $E=-\wp(t)$.
\end{prop}
 
\section*{Acknowledgments} 
We are grateful to Boris Shapiro for an initiative that led to this collaboration. This work was supported by the Swedish Science Research Council (VR), the G\"oran Gustafsson Foundation and the Japan Society for the Promotion of Science.

\appendix
\section{Elliptic functions}
\label{appA} 

\subsection{Definitions} 
\label{appA1} 
The  Weierstrass $\wp$-function with periods $(2\omega_1, 2\omega_3)$ is defined as follows:
\begin{equation}
 \wp (x)= \frac{1}{x^2}+  \sum_{(m,n)\in \Z \times \Z \setminus \{ (0,0)\} } \left\{ \frac{1}{(x-\Omega_{m,n})^2}-\frac{1}{\Omega_{m,n}^2}\right\}
\end{equation} 
with $\Omega_{m,n}=2m\omega_1 +2n\omega_3$. We also recall the definitions of the corresponding Weierstrass zeta- and sigma-functions, 
\begin{equation*}
 \zeta (x)= \frac{1}{x}+  \sum_{(m,n)\in \Z \times \Z \setminus \{ (0,0)\} } \left\{ \frac{1}{x-\Omega_{m,n}}+ \frac1{\Omega_{m,n}}+\frac{z}{\Omega_{m,n}^2}\right\}
\end{equation*} 
and 
\begin{equation*}
\sigma (x)=x\prod_{(m,n)\in \Z \times \Z \setminus \{(0,0)\} }\left\{ \left(1-\frac{x}{\Omega_{m,n}}\right)
\exp\left(\frac{x}{\Omega_{m,n}}+\frac{x^2}{2\Omega_{m,n}^2}\right)\right\},
\end{equation*} 
respectively. 

We use the following symbols, 
\begin{equation} 
\omega _0=0, \quad \omega_2=-\omega_1-\omega_3,
\end{equation} 
\begin{equation}
 e_{\nu } =\wp(\omega_{\nu }), \; \; \; \eta_{\nu } =\zeta(\omega_{\nu }) \; \; \; \; (\nu =1,2,3), 
\label{enuetanu} 
\end{equation}
\begin{equation}
q=\ee{\pi \ii \tau },\quad \tau =\frac{\omega _3}{\omega _1},\quad R = \frac{2\omega_1}{\pi},\quad \beta= \frac{2\omega_1\omega_3}{\pi\ii}
\end{equation} 
where we regard $\omega_1$ as fixed.  We also need $\theta_\nu(x)=\vartheta_\nu(x/R,q)$ ($\nu=1,2,3,4$)  with the Theta-functions $\vartheta_\nu(x , q)$  as usual \cite{WW}, i.e., 
\begin{equation}
\label{thenu} 
\begin{split} 
\theta _1(x) &= 2\sum _{n=0}^{\infty } (-1)^{n} \ee{\tau \pi \ii (n+1/2)^2} \sin (2n+1)x/R, \\
\theta _2(x) &= 2\sum _{n=0}^{\infty } \ee{\tau \pi \ii (n+1/2)^2} \cos (2n+1)x/R ,\\ 
\theta _3 (x) &= 1+ 2\sum _{n=1}^{\infty } \ee{\tau \pi \ii n ^2} \cos 2n x/R ,\\
\theta _4 (x) &= 1+ 2\sum _{n=1}^{\infty } (-1)^{n-1} \ee{\tau \pi \ii n ^2} \cos 2n x/R.
\end{split}
\end{equation} 
We also define
\begin{equation} 
\label{phinu} 
\phi_{\nu}(x)=\frac{\theta_{\nu}'(x)}{\theta_{\nu}(x)} \; \; \; \; (\nu=1,2,3,4).  
\end{equation} 
Here and in the following we use the following shorthand notation,
\begin{equation} 
\theta_\nu'(x)\equiv \frac{\partial}{\partial x}\theta_\nu(x),\; \; \; \dot\theta_\nu(x)\equiv \frac{\partial}{\partial \beta}\theta_\nu(x) = \frac{2\pi \ii}{\omega_1^2}\frac{\partial}{\partial\tau}\tet_{\nu}(x).
\end{equation} 
 
\subsection{Properties} 
\label{appA2} 
We recall some well-known properties of the Weierstrass elliptic functions that we need (see e.g.\ \cite{WW}, Chapter XX): 
\begin{equation} 
\wp(x)=-\zeta'(x),\; \; \;  \zeta(x)=\frac{\sigma'(x)}{\sigma(x)}, 
\label{wpxi}
\end{equation} 
\begin{equation}
e_1+e_2+e_3 =\eta_1+\eta_2+\eta_3=0 , 
\label{sumrule} 
\end{equation}
\begin{equation}
\wp(x+2\omega_{\nu })=\wp(x), \; \; \; \zeta(x+2\omega_{\nu })=\zeta(x)+2\eta_{\nu } \; \; \; \; (\nu =1,2,3) 
\label{periods} 
\end{equation}
\begin{equation} 
\wp(-x)=\wp(x),\quad \xi(-x)=-\xi(x),\quad \sigma(-x)=-\sigma(x),
\label{Z2} 
\end{equation} 
and
\begin{align}
\bigl(\zeta(x_1) +\zeta (x_2) + \zeta (x_3)\bigr)^2= \wp (x_1) +\wp (x_2) +\wp (x_3) \; \; \; \; (x_1+x_2+x_3=0).  \label{eq:zetawp}
\end{align}
As will be seen, the last identity plays an important role in the proof of our main result. Other identities that we need are the heat equation satisfied by the Theta-functions, i.e., 
\begin{equation}
\theta''_{\nu} (x ) =   2 \dot\theta_{\nu}(x)\; \; \; \;  (\nu=1,2,3,4)
\label{heateq} 
\end{equation} 
(obvious from the definitions), and two identities obtained from the well-known duplication formula   
\begin{equation*}
\theta_1 (2x)= C  \theta _{1} (x ) \theta _{2} (x) \theta _{3} (x) \theta _{4} (x ) 
\end{equation*}
where $C$ is a constant (see e.g.\ \cite{WW}, Example~5 at the end of Chapter XXI), i.e., 
\begin{equation*}
\frac{\theta_{1}'(2x)}{\theta_{1} (2x)} = \frac{1}{2} \left( \frac{\theta _{1} '(x ) }{\theta _{1} (x) } + \frac{\theta _{2} '(x ) }{\theta _{2} (x) } + \frac{\theta _{3} '(x ) }{\theta _{3} (x ) } + \frac{\theta _{4} '(x ) }{\theta _{4} (x ) } \right) 
\end{equation*} 
and
\begin{equation*} 
\frac{\theta_{1} ''(2x)}{\theta_{1}(2x)} = \frac{1}{4} \sum _{\nu =1}^4  \frac{\theta _{\nu } ''(x ) }{\theta _{\nu } (x ) } + \frac{1}{2} \sum _{1 \leq \mu < \nu \leq 4 }  \frac{\theta _{\mu } '(x ) }{\theta _{\mu } (x) }\frac{\theta _{\nu } '(x ) }{\theta_{\nu } (x ) } .
\end{equation*} 
Using \Ref{phinu} and the heat equation we can write these identities as  
\begin{equation} 
\phi_1(2x)=\frac12\sum_{\nu=1}^4\phi_\nu(x) 
\label{double1} 
\end{equation} 
and 
\begin{equation} 
\frac{\dot\theta_1(2x)}{\theta_1(2x)} = \frac14 \sum_{\nu=1}^4\frac{\dot\theta_\nu(x)}{\theta_\nu(x)} + \frac14\sum_{1\leq\mu<\nu\leq 4} \phi_\mu(x)\phi_\nu(x) . \label{double2} 
\end{equation}  

We also need the relations between the functions defined in \Ref{phinu} and the Weierstrass  zeta function:
\begin{equation} 
\phi_1(x) = \zeta(x) -\frac{\eta_1 x}{\omega_1} , \; \; \; \phi_{\nu+1}(x) = \zeta(x+\omega_\nu) -\eta_\nu -\frac{\eta_1 x}{\omega_1} \; \; \; \;  (\nu=1,2,3)
\label{zetanuphinu}
\end{equation} 
(this is a simple consequence of 
\begin{equation*} 
\theta_{1}(x) = C_0 \sigma(x)\exp\Bigl(-\frac{\eta_1 x^2}{2\omega_1} \bigr),\; \; \; 
\theta_{\nu+1}(x) = C_\nu \sigma(x+\omega_\nu)\exp\left(-\eta_\nu x -\frac{\eta_1 x^2}{2\omega_1} \right)
\; \; \; \;  (\nu=1,2,3)
\end{equation*} 
for constants $C_\nu$ ($\nu=0,1,2,3$), which can be obtained by comparing the product representation of the Weierstrass sigma function in \cite{WW}, \S\;20.421, with the product representations of the Theta-functions in \cite{WW}, \S\;21.3). 

\bigskip

In the following we collect several identities needed in the proof of our main result. 

\begin{prop}
The following holds true, 
\begin{equation}
\label{I1} 
\phi_{\nu+1}'(x) = - \wp(x+\omega_\nu) -\frac{\eta_1}{\omega _1}
\end{equation} 
and 
\begin{equation}
\label{I2} 
\phi_{\nu+1}(x)^2 =2\frac{\dot\tet_{\nu+1}(x)}{\tet_{\nu+1}(x)} + \wp(x+\omega_\nu) + \frac{\eta_1}{\omega _1}
\end{equation} 
for $\nu=0,1,2,3$, and 
\begin{equation}
\label{I3} 
\phi_{\nu+1}(x)\phi_{\mu+1}(x) = \frac{\dot\tet_{\nu+1}(x)}{\tet_{\nu+1}(x)}+\frac{\dot\tet_{\mu+1}(x)}{\tet_{\mu+1}(x)} +\frac{\eta_1}{\omega _1} -\frac{e_{\nu,\mu}}{2}
\end{equation} 
for $0\leq \mu<\nu\leq 3$, where 
\begin{equation} 
e_{\nu ,0} =e _{\nu }\quad (\nu=1,2,3), \quad  e_{2,1}=e_3,\quad  e_{3,1} =e_2,\quad  e_{3,2} =e_1.
\label{emunu} 
\end{equation} 
Moreover, 
\begin{equation}
\label{I4} 
\begin{split} 
\phi_1(x-y)\phi_1(x+y) = \frac12\sum_{\nu=0}^3\sum_{r=\pm} \phi_{\nu+1}(x)\phi_1(x-r y) - \sum_{\nu=0}^3\frac{\dot\tet_{\nu+1}(x)}{\tet_{\nu+1}(x)}\\ -\sum_{r=\pm}\frac{\dot\tet_1(x-ry)}{\tet_1(x-ry)}-\frac{3\eta_1}{\omega _1} 
\end{split} 
\end{equation} 
and 
\begin{equation}
\label{I5} 
\sum_{r=\pm}\bigl( \phi_{\nu+1}(x)-r\phi_{\nu+1}(y)\bigr) \phi(x-ry) = \sum_{r=\pm}\Bigl( \frac{\dot\tet_{\nu+1}(x)}{\tet_{\nu+1}(x)}+\frac{\dot\tet_{\nu+1}(y)}{\tet_{\nu+1}(y)}+\frac{\dot\tet_1(x-ry)}{\tet_1(x-ry)}+\frac{3\eta_1}{2\omega _1}\Bigl) 
\end{equation} 
for $\nu=0,1,2,3$, and 
\begin{equation}
\label{I6} 
\begin{split} 
\sum_{r,s=\pm} \Bigl(\phi_1(x-ry)\phi_1(x-sz) + \phi_1(y-rx)\phi_1(y-sz) +\phi_1(z-rx)\phi_1(z-sy)  \Bigr) \\ = 2\sum_{r=\pm} \Bigl( \frac{\dot\tet_1(x-ry)}{\tet_1(x-ry)}+  \frac{\dot\tet_1(x-rz)}{\tet_1(x-rz)} +  \frac{\dot\tet_1(y-rz)}{\tet_1(y-rz)} +\frac{3\eta_1}{2\omega _1} \Bigr) .
\end{split} 
\end{equation} 
\end{prop}
\begin{proof}
Differentiate \Ref{zetanuphinu} in $x$ and use \Ref{wpxi} to obtain \Ref{I1}. 

Differentiate \Ref{phinu} in $x$ and use \Ref{phinu} again to obtain $\phi_{\nu}(x)^2 = \theta''_{\nu}(x)/\theta_{\nu}(x) - \phi_{\nu}'(x)$. Insert the heat equation in \Ref{heateq} and use \Ref{I1} to obtain \Ref{I2}. 

We show \Ref{I3}. Substitute $x_1= x +\omega _{\nu }$, $x_2=-x- \omega _{\mu }$ and $x_3=\omega _{\mu }-\omega _{\nu }$ in \Ref{eq:zetawp}, and use $\zeta (\omega _{\mu }-\omega _{\nu } ) = \eta _{\mu } -\eta _{\nu} $ and $\wp (\omega _{\mu }-\omega _{\nu } ) = e_{\nu,\mu}$ to obtain 
\begin{equation*}
 \bigl( \zeta (x+\omega _\nu ) -\zeta (x +\omega _\mu ) +\eta _{\mu } -\eta _{\nu}  \bigr)^2 = \wp (x +\omega _{\nu }) +\wp (x+\omega _{\mu }) + e_{\nu,\mu}.
\end{equation*}
Insert into this 
\begin{equation*}
\zeta (x+\omega _\nu ) -\zeta (x +\omega _\mu ) +\eta _{\mu } -\eta _{\nu} = \phi _{\nu +1}(x ) - \phi _{\mu +1} (x )  
\end{equation*} 
(this follows from \Ref{zetanuphinu}), expand the square, insert \Ref{I2}, and obtain an identity equivalent to \Ref{I3}. 

We show \Ref{I4}. Substitute $x_1=x-y$, $x_2=x+y$, $x_3=-2x$ in \Ref{eq:zetawp}, and use  \Ref{phinu} and \Ref{zetanuphinu} to obtain 
\begin{equation*}
\big(\phi_1(x-y) + \phi_1(x+y)-\phi_1(2x)\bigr)^2 = \wp(x-y)+\wp(x+y)+\wp(2x).
\end{equation*}   
Expanding the square and using \Ref{I2} this can be written as 
\begin{equation*} 
\phi_1(x-y)\phi_1(x+y) = \phi_1(2x)\sum_{r=\pm}\phi_1(x-ry) -\frac{\dot\theta_1(2x)}{\theta_1(2x)}-\sum_{r=\pm}\frac{\dot\theta_1(x-ry)}{\theta_1(x-ry)}-\frac{3\eta_1}{2\omega_1}.
\end{equation*} 
Insert into this \Ref{double1} and 
\begin{equation*}
\begin{split}  
\frac{\dot\theta_1(2x)}{\theta_{1}(2x)} = \frac{1}{4}\sum_{\nu=0}^3\frac{\dot\theta_{\nu +1}(x)}{\theta_{\nu +1}(x)}+\frac{1}{4}\sum_{0\leq\mu<\nu\leq 3}\Bigl(\frac{\dot\tet_{\mu +1}(x)}{\tet_{\mu +1}(x)}+\frac{\dot\tet_{\nu +1}(x)}{\tet_{\nu +1}(x)} +\frac{\eta_1}{\omega _1} -\frac{e_{\nu,\mu}}{2}
 \Bigr) 
 \\= \sum_{\nu=0}^3\frac{\dot\theta_{\nu +1}(x)}{\theta_{\nu +1}(x)} +\frac{3\eta_1}{2\omega_1}
\end{split} 
\end{equation*} 
(we used \Ref{double2}, \Ref{I3} and $\sum_{\mu<\nu}e_{\nu,\mu}=2(e_1+e_2+e_3)=0$) to obtain \Ref{I4}.

We show \Ref{I5}. We substitute $x_1= x+\omega _{\nu }$, $x_2=-y-\omega _{\nu }$ and $x_3=-x+y$ in \Ref{eq:zetawp}. Then 
\begin{align*}
 \bigl( \phi_{\nu +1} (x ) - \phi_{\nu +1} (y )  -  \phi _{1} (x-y ) \bigr) ^2 &= \bigl( \zeta (x+\omega _{\nu} ) -\zeta (y+\omega _{\nu}) -\zeta (x-y)\bigr)^2 \\ &= \wp (x +\omega _{\nu }) +\wp (y+\omega _{\nu }) +\wp (x-y)
\end{align*}
using \Ref{zetanuphinu}, and it follows from \Ref{I2} that
\begin{align*}
& \phi _{\nu +1} (x) \phi _{\nu +1} (y) + \phi _{\nu +1} (x) \phi_1(x-y ) -\phi _{\nu +1} (y ) \phi_1(x-y )\\
& = \frac{\dot\tet_{\nu+1}(x)}{\tet_{\nu+1}(x)}+\frac{\dot\tet_{\nu+1}(y)}{\tet_{\nu+1}(y)}+\frac{\dot\tet_1(x-y)}{\tet_1(x-y)}+\frac{3\eta_1}{2\omega _1}. \nonumber
\end{align*}
Similarly we have
\begin{align*}
& -\phi _{\nu +1} (x) \phi _{\nu +1} (y) + \phi _{\nu +1} (x) \phi_1(x+y ) +\phi _{\nu +1} (y ) \phi_1(x+y )\\
& = \frac{\dot\tet_{\nu+1}(x)}{\tet_{\nu+1}(x)}+\frac{\dot\tet_{\nu+1}(y)}{\tet_{\nu+1}(y)}+\frac{\dot\tet_1(x+y)}{\tet_1(x+y)}+\frac{3\eta_1}{2\omega _1}, \nonumber
\end{align*}
by substituting $x_1= x+\omega _{\nu }$, $x_2=y+\omega _{\nu }$ and $x_3=-x-y-2\omega _{\nu }$ in \Ref{eq:zetawp}. Sum up the last two equalities and obtain \Ref{I5}. 

We finally show \Ref{I6}. We substitute $x_1= x+ry$, $x_2=-x-sz$ and $x_3=-ry+sz$, for $r,s=\pm$,  in \Ref{eq:zetawp}. Then 
\begin{equation*} 
\bigl(\phi_1(x+ry) - \phi_1(x+sz) - \phi_1(ry-sz)\bigr)^2 =\wp(x+ry)+\wp(x+sz)+\wp(y-rsz)
\end{equation*} 
and, by using \Ref{I2} as above,  
\begin{equation*}
\begin{split} 
\phi_1(x+ry)\phi_1(x+sz) + \phi_1(y+rx)\phi_1(y-rsz) + \phi_1(z+sx)\phi_1(z-rsy) \\ 
= \frac{\dot\theta_1(x+ry)}{\theta_1(x+ry)} + \frac{\dot\theta_1(x+sz)}{\theta_1(x+sz)} + \frac{\dot\theta_1(y-rsz)}{\theta_1(y-rsz)}+ \frac{3\eta_1}{2\omega_1} .
\end{split} 
\end{equation*} 
Sum the last equality over $r,s=\pm$ to obtain \Ref{I6}. 
\end{proof}

\section{Proof of Source Identity (details)}
\label{appB} 
We compute the functions defined in \Ref{cQ}, 
\begin{equation}
\cV_J = \sum_{\nu} g_{\nu,J}\phi_{\nu+1}(X_J) +\sum_{K\neq J} \sum_{r=\pm}m_Jm_K\lambda \phi_1(X_J-rX_K) , 
\end{equation}
and thus 
\begin{equation}
\label{tcH} 
\tilde{\cH}= -\sum_J \frac1{m_J}\partial_J^2 + \cW,\quad \cW =\sum_J \frac1{m_J}\Bigl( \partial_J\cV_J + \cV_J^2\Bigr)
\end{equation}
with $\partial_J=\partial/\partial X_J$. Here and in the following, we write $\sum_\nu$ short for $\sum_{\nu=0}^3$, $\sum_J$ short for $\sum_{J=1}^{\cN}$, etc. 

We compute $\cW = \cW_1+\cW_2+\cW_3$ with 
\begin{equation}
\begin{split} 
\label{cW1} 
\cW_1 = \sum_J \frac1{m_J}\Biggl\{ \sum_{\nu} \Bigl( g_{\nu,J}\phi'_{\nu+1} (X_J)  + g_{\nu,J}^2\phi_{\nu+1}(X_J)^2\Bigr) \\ + 2\sum_{\nu<\mu}g_{\nu,J}g_{\mu,J}\phi_{\nu+1}(X_J)\phi_{\mu+1}(X_J)\Biggr\} 
\end{split} 
\end{equation} 
the sum of all one-body terms, 
\begin{equation}
\label{cW2}
\begin{split} 
\cW_2 = \sum_J \sum_{K\neq J} \Biggl\{   \sum_{r=\pm} \Biggl(m_K\lambda\phi'_1(X_J-rX_K) + m_Jm_K^2\lambda^2 \phi_1(X_J-rX_K)^2  \\ 
+ 2\sum_{\nu}g_{\nu,J} m_K\lambda\phi_{\nu+1}(X_J)\phi_1(X_J-rX_K)\Biggr)\\ 
+ 2 m_Jm_K^2\lambda^2 \phi_1(X_J-X_K)\phi_1(X_J+X_K) \Biggr\}
\end{split} 
\end{equation} 
the sum of all two-body terms, and
\begin{equation}
\label{cW3} 
\cW_3 = \sum_J\sum_{K\neq J}\sum_{L\neq J,K} \sum_{r,s=\pm }m_Jm_Km_L\lambda^2 \phi_1(X_J-rX_K)\phi_1(X_J-sX_L) 
\end{equation} 
the sum of all three-body terms. 

To simplify the one-body terms we use \Ref{I1}--\Ref{I3} and obtain 
\begin{equation*}
\begin{split} 
\cW_1= \sum_J \frac1{m_J}\Biggl\{ \sum_{\nu}\left( g_{\nu,J}(g_{\nu,J}-1)\left( \wp(z+\omega_\nu) +\frac{\eta_1}{\omega_1} \right) + 2g_{\nu,J}^2\frac{\dot\tet_{\nu+1}(X_J)}{\tet_{\nu+1}(X_J)}\right) \\  + 2\sum_{\mu<\nu} g_{\nu,J} g_{\mu,J} \left(\frac{\dot\tet_{\nu+1}(X_J)}{\tet_{\nu+1}(X_J)}+\frac{\dot\tet_{\mu+1}(X_J)}{\tet_{\mu+1}(X_J)} +\frac{\eta_1}{\omega_1}-\frac{e_{\nu,\mu}}{2} \right) \Biggr\} .
\end{split} 
\end{equation*} 
Changing summations in the last sum we obtain, after some computations,
\begin{equation}
\label{W1r} 
\begin{split} 
\cW_1 = \sum_J \sum_{\nu}\Biggl( \frac1{m_J} g_{\nu,J}(g_{\nu,J}-1)\wp(z+\omega_\nu) + 2g_{\nu,J}(|d|+2m_J\lambda) \frac{\dot\tet_{\nu+1}(X_J)}{\tet_{\nu+1}(X_J)} \Biggl) 
\\ + \sum_J \left( (|d|+2m_J\lambda)(m_J |d| + 2m_J^2\lambda-1)\frac{\eta_1}{\omega_1}- m_J\sum_{\mu<\nu}d_\nu d_\mu e_{\nu,\mu}\right) 
\end{split} 
\end{equation} 
where we used
\begin{equation*} 
 \sum_{\mu}g_{\mu,J} = m_J(|d|+2m_J\lambda),\quad \sum_{\mu<\nu}e_{\nu,\mu}=0,\quad \sum_{\mu<\nu} (d_\mu+d_\nu)e_{\nu,\mu}=0
\end{equation*} 
following from \Ref{gnuJ},  \Ref{sumrule} and \Ref{emunu}. 

To compute the two-body terms we use \Ref{I1}, \Ref{I2} and \Ref{I4} to obtain
\begin{equation*} 
\begin{split} 
\cW_2 = \sum_J \sum_{K\neq J} \Biggl\{   \sum_{r=\pm} \Biggl(m_K\lambda(m_Jm_K\lambda-1)\left(\wp(X_J-rX_K) + \frac{\eta_1}{\omega_1} \right) \\ + \sum_{\nu}( 2g_{\nu,J} m_K\lambda + m_Jm_K^2\lambda^2) \phi_{\nu+1}(X_J)\phi_1(X_J-rX_K)\Biggr)\\ -2m_Jm_K^2\lambda^2\sum_\nu \left(\frac{\dot\theta_{\nu+1}(X_J)}{\theta_{\nu+1}(X_J)} + \frac{3\eta_1}{4\omega_1}\right) \Biggr\}
\end{split} 
\end{equation*} 
where terms $\sum_{r=\pm} \dot\theta_1(X_J-rX_K)/\theta_1(X_J-rX_K)$ are cancelled. Symmetrizing  the terms in the first two lines using $(m_K+m_J)\lambda(m_Jm_K\lambda-1)=\gamma_{JK}$, 
\begin{equation*}
2g_{\nu,J} m_K\lambda + m_Jm_K^2\lambda^2 = m_Jm_K\lambda \bigl(2d_\nu + (m_J+m_K)\lambda\bigr)
\end{equation*} 
and 
\begin{equation*}
\phi_1(X_K-rX_J) = -r\phi_1(X_J-rX_K), 
\end{equation*} 
we can write this as
\begin{equation*}
\begin{split} 
\cW_2 = &\sum_{J<K} \sum_{r=\pm} \Biggl\{ \gamma_{JK} \left(\wp(X_J-rX_K) + \frac{\eta_1}{\omega_1} \right) \\ &+ \sum_{\nu}m_Jm_K\lambda \bigl(2d_\nu + (m_J+m_K)\lambda\bigr)\bigl(\phi_{\nu+1}(X_J)-r\phi_{\nu+1}(X_K)\bigr) \phi_1(X_J-rX_K)\Biggr\}\\ & -\sum_J\sum_{K\neq J}2m_J m_K^2\lambda^2\sum_{\nu}\left( \frac{\dot\theta_{\nu+1}(X_J)}{\theta_{\nu+1}(X_J)} + \frac{3\eta_1}{4\omega_1}\right).
\end{split} 
\end{equation*} 
We insert \Ref{I5} and partly undo the symmetrization, 
\begin{equation*}
\begin{split} 
\cW_2 = &\sum_{J<K} \sum_{r=\pm} \Biggl\{ \gamma_{JK} \left(\wp(X_J-rX_K) + \frac{\eta_1}{\omega_1} \right) \\ & + \sum_{\nu}m_Jm_K\lambda \bigl(2d_\nu + (m_J+m_K)\lambda\bigr)\frac{\dot\theta_1(X_J-rX_K)}{\theta_1(X_J-rX_K)}\Biggr\} \\ & + \sum_J\sum_{K\neq J}\sum_{\nu}\Bigl(2m_Jm_K\lambda \bigl(2d_\nu + (m_J+m_K)\lambda\bigr)- 2m_Jm_K^2\lambda^2 \Bigr)\left( \frac{\dot\theta_{\nu+1}(X_J)}{\theta_{\nu+1}(X_J)}  + \frac{3\eta_1}{4\omega_1}\right) 
\end{split} 
\end{equation*} 
(a factor 2 in the first term in the last line comes from a $r$-sum) and, after some computations, we obtain 
\begin{equation} 
\label{W2r} 
\begin{split} 
\cW_2 = &\sum_{J<K} \sum_{r=\pm} \Biggl\{ \gamma_{JK} \wp(X_J-rX_K) + m_Jm_K\lambda(2|d|+4(m_J+m_K)\lambda) \frac{\dot\theta_1(X_J-rX_K)}{\theta_1(X_J-rX_K)}\Biggr\} \\ & + \sum_J\sum_{\nu}4\lambda (|m|-m_J)g_{\nu,J} \frac{\dot\theta_{\nu+1}(X_J)}{\theta_{\nu+1}(X_J)}  \\ &+ \left(\sum_{J<K}2\gamma_{JK} + \sum_J3\lambda (|m|-m_J) m_J (|d|+2m_J\lambda) \right)\frac{\eta_1}{\omega_1}
\end{split}
\end{equation} 
using the short-hand notation in \Ref{nm}. 

To compute the three-body terms we symmetrize the summations, 
\begin{equation*}
\begin{split} 
\cW_3 = &\sum_{J<K<L}m_Jm_Km_L\lambda^2 
 \sum_{r,s=\pm}\Bigl(\phi_1(X_J-rX_K)\phi_1(X_J-sX_L) \\ &+ \phi_1(X_K-rX_L)\phi_1(X_K-sX_J) +\phi_1(X_L-rX_J)\phi_1(X_L-sX_K)  \Bigr) , 
\end{split} 
\end{equation*} 
which allows us to use the identity \Ref{I6} to obtain 
\begin{equation*} 
\begin{split} 
\cW_3 = \sum_{J<K<L} 2 m_J m_K m_L\lambda^2  \sum_{r=\pm}\Biggl(\frac{\dot\tet_1(X_J-rX_K)}{\tet_1(X_J-rX_K)}+  \frac{\dot\tet_1(X_J-rX_L)}{\tet_1(X_J-rX_L)} +  \frac{\dot\tet_1(X_K-rX_L)}{\tet_1(X_K-rX_L)}+\frac{3\eta_1}{2\omega_1}  \Biggr) .
\end{split} 
\end{equation*} 
Changing summations again we can write this as 
\begin{equation*}
\cW_3=\sum_J\sum_{K\neq J}\sum_{L\neq J,K}  2 m_J m_K m_L\lambda^2  \sum_{r=\pm}\left( \frac{\dot\tet_1(X_J-rX_K)}{\tet_1(X_J-rX_K)}+ \frac{\eta_1}{2\omega_1}\right) ,
\end{equation*} 
and, by simple computations, we obtain 
\begin{equation}
\label{W3r} 
\cW_3 = \sum_{J<K} 4 m_Jm_K(|m|-m_J-m_K)\lambda^2  \sum_{r=\pm}\left( \frac{\dot\tet_1(X_J-rX_K)}{\tet_1(X_J-rX_K)}+ \frac{\eta_1}{2\omega_1}\right) . 
\end{equation} 

Recalling \Ref{tcH} and $\cW = \cW_1+\cW_2+\cW_3$ we add the results in \Ref{W1r}, \Ref{W2r} and \Ref{W3r} and obtain  
\begin{equation*}
\begin{split}
\tilde{\cH} =& \sum_J \frac1{m_J}\left(-\partial_J^2 + \sum_{\nu} g_{\nu,J}(g_{\nu,J}-1)\wp(X_J+\omega_\nu)  \right) +\sum_{J<K}\sum_{r=\pm}\gamma_{J,K}\wp(X_J-rX_K) \\ & +(|d|+2|m|\lambda)\left( \sum_J\sum_{\nu} 2g_{\nu,J}\frac{\dot\theta_{\nu+1}(X_J)}{\theta_{\nu+1}(X_J)}+ \sum_{J<K}\sum_{r=\pm} m_Jm_K\lambda  \frac{\dot\tet_1(X_J-rX_K)}{\tet_1(X_J-rX_K)} \right) -\cE_0
\end{split} 
\end{equation*} 
with
\begin{equation}
\label{cE01}
\begin{split} 
\cE_0 = -\Biggl\{ \sum_J\Bigl((|d|+2m_J\lambda)(m_J|d|+2m_J^2\lambda-1) +3\lambda(|m|-m_J)m_J(|d|+2m_J\lambda) \Bigr)\\ + \sum_{J<K}\Bigl(2\gamma_{JK} + 4m_Jm_K(|m|-m_J-m_K)\lambda^2 \Bigr)  \Biggr\}\frac{\eta_1}{\omega_1} + \sum_Jm_J \sum_{\mu<\nu}d_\nu d_\mu e_{\mu,\nu} 
\end{split} 
\end{equation} 
(a factor 2 in the second term in the second line is from a $r$-sum). Recalling \Ref{Phi0} and \Ref{cH} we can write this as in \Ref{result}. Some computations show that $\cE_0$ in \Ref{cE01} can be simplified to the formula given in \Ref{cE0}.


\begin{thebibliography}{9999}

\bibitem{C} F. Calogero, \emph{Solution of the one-dimensional {N}-body problems with  quadratic and/or inversely quadratic pair potentials}, J. Math. Phys.\textbf{12} (1971) 419--436

\bibitem{M} J. Moser , \emph{Three integrable Hamiltonian systems connected with isospectral deformations}, Adv. Math. 16 (1975) 1--23

\bibitem{Su} B. Sutherland: \emph{Exact results for a quantum many body problem in one-dimension. II.}, Phys. Rev. \textbf{A5} (1972) 1372--1376

\bibitem{OP} M.A. Olshanetsky and A.M. Perelomov, \emph{Quantum completely integrable systems connected with semisimple Lie algebras}, Lett. Math. Phys. \textbf{2} (1977) 7--13

\bibitem{BF} T.H. Baker and P.J. Forrester, \emph{The Calogero-Sutherland model and generalized classical polynomials}, Commun. Math. Phys. \textbf{188} (1997)  175--216

\bibitem{HL} M. Halln\"as and E. Langmann, \emph{A unified construction of generalised classical polynomials associated with operators of Calogero-Sutherland type}, Constr. Approx. \textbf{31} (2010) 309--342

\bibitem{I} V.I. Inozemtsev, \emph{Lax representation with spectral parameter on a torus for integrable particle systems}, Lett. Math. Phys. \textbf{17} (1989) 11--17

\bibitem{vD} J.F. van Diejen, \emph{Integrability of difference Calogero-Moser systems},  J. Math. Phys. \textbf{35} (1994) 2983--3004

\bibitem{O} T. Oshima, \emph{Completely integrable systems with a symmetry in coordinates}. Asian J. Math. \textbf{2} (1998) 935--955 

\bibitem{Ron} A. Ronveaux (ed): Heun's differential equations. Oxford Science Publications, Oxford University Press, Oxford (1995) 

\bibitem{SL} S. Slavyanov and W. Lay:  Special Functions., Oxford Science Publications, Oxford University Press, Oxford (2000) 

\bibitem{TakID} K. Takemura, \emph{Integral transformation and Darboux transformation of Heun's differential equation}. 
{in: Nonlinear and modern mathematical physics,} AIP Conference Proceedings {\bf 1212} (2010) 58--65

\bibitem{KSY} H. Kihara, M. Sakaguchi, Y. Yasui, \emph{Scalar Laplacian on Sasaki-Einstein manifolds $Y\sp {p,q}$},  Phys. Lett. B {\bf 621} (2005) 288--294

\bibitem{WW} E.T. Whittaker and G.N. Watson: A course of modern analysis, Fourth Edition, Cambridge University Press (1927) 

\bibitem{EK} P. I. Etingof and A. A. Kirillov, Jr., \emph{Representations of affine Lie algebras, parabolic differential equations, and Lame functions}, Duke Math. J. \textbf{74} (1994) 585--614

\bibitem{BM} V.V. Bazhanov and V.V. Mangazeev, \emph{Eight-vertex model and non-stationary Lame equation}, J. Phys. A \textbf{38} (2005) L145--L153 

\bibitem{FLNO} V.A. Fateev, A.V.  Litvinov, A. Neveu, E. Onofri, \emph{A differential equation for a four-point correlation function in Liouville field theory and elliptic four-point conformal blocks}, J. Phys. A \textbf{42} (2009) 304011 (29 pages)  

\bibitem{Rui3} S.N.M. Ruijsenaars, \emph{Hilbert-Schmidt operators vs. integrable systems of elliptic Calogero-Moser type III. The Heun case},  SIGMA {\bf 5} (2009) 049 (21 pages)

\bibitem{Tak1} K. Takemura,  \emph{The Heun equation and the Calogero-Moser-Sutherland system I: the Bethe Ansatz method}, Comm. Math. Phys. {\bf 235} (2003) 467--494

\bibitem{Tak3} K. Takemura, \emph{The Heun equation and the Calogero-Moser-Sutherland system III: the finite gap property and the monodromy},  J. Nonlin. Math. Phys. {\bf 11} (2004) 21--46 

\bibitem{Tak4} K. Takemura, \emph{The Heun equation and the Calogero-Moser-Sutherland system IV: the Hermite-Krichever Ansatz}, Comm. Math. Phys. {\bf 258} (2005) 367--403

\bibitem{CFV} O. Chalykh, M. Feigin and A. Veselov, \emph{New integrable generalizations of  Calogero-Moser quantum problems}, J. Math. Phys. \textbf{39} (1998) 695--703

\bibitem{Sergeev} A.N. Sergeev, \emph{Calogero operator and Lie superalgebras}, Theor. Math. Phys. \textbf{131} (2002) 747--764

\bibitem{SV} A.N. Sergeev and A. Veselov, \emph{Deformed quantum Calogero-Moser systems and Lie superalgebras}, Commun. Math. Phys. \textbf{245} (2004) 249--278

\bibitem{EL0} E. Langmann, \emph{Algorithms to solve the Sutherland model}, J. Math. Phys. \textbf{42} (2001 4148--4157

\bibitem{EL} E. Langmann, \emph{A method to derive explicit formulas for an elliptic generalization of the Jack polynomials} in: \emph{Jack, Hall-Littlewood and Macdonald polynomials}, V.B. Kuznetsov and S. Sahi (eds.), Contemporary Mathematics, American Mathematical Society (2006) 257--270 

\bibitem{Sen} D. Sen, \emph{A multispecies Calogero-Sutherland model}, Nucl. Phys. B \textbf{479} (1996) 554--574

\bibitem{EL1} E. Langmann, \emph{Singular eigenfunctions of Calogero-Sutherland type systems and how to transform them into regular ones},  SIGMA 3 (2007) 031 (18 pages)  

\bibitem{EL2} E. Langmann, \emph{Source identity and kernel functions for elliptic Calogero-Sutherland type systems}, Lett. Math. Phys. \textbf{94} (2010) 63--75

\bibitem{Woj} S. Wojciechowski: \emph{The analogue of the B\"acklundtransformation for integrable many-body systems},  J. Phys. A: Math. Gen. \textbf{15} (1982) L653--L657

\bibitem{KS} V.B. Kuznetsov and E.K. Sklyanin: \emph{On B\"acklund transformations for many-body systems},  J. Phys. A: Math. Gen. \textbf{31} (1998) 2241--2251

\bibitem{KMS} V.B. Kuznetsov, V.V. Mangazeev, and E.K. Sklyanin: \emph{Q-operator and factorised separation chain for Jack polynomials}, Indag. Math. \textbf{14} (2003) 451--482

\bibitem{Skly} E.K. Sklyanin: \emph{Separation of variables. New trends}, Prog. Theor. Phys. Suppl. \textbf{118} (1995) 35--60

\bibitem{GGR} D. G\'omez-Ullate, A. Gonz\'alez-L\'opez, A. Rodr\'iguez: \emph{Exact solutions of an elliptic Calogero Sutherland model},  Phys. Lett. B \textbf{511} (2001) 112--118

\bibitem{T} K. Takemura: \emph{Quasi-exact solvability of Inozemtsev models}, J. Phys. A: Math. Gen. \textbf{35} (2002) 8867-8881

\bibitem{R} S.N.M. Ruijsenaars: \emph{Complete integrability of relativistic Calogero-Moser systems and elliptic function identities}, Comm. Math. Phys. \textbf{110} (1987) 191-213

\bibitem{KNS} Y. Komori, M. Noumi and J. Shiraishi, \emph{Kernel functions for difference operators of Ruijsenaars type and their applications}, SIGMA \textbf{5} (2009), 054 (40 pages) 

\bibitem{RS2} M. Reed and B. Simon: Methods of Modern Mathematical Physics. II: Fourier Analysis, Self-Adjointness. Academic Press, New York (1975) 

\end{thebibliography}
\end{document}